\documentclass[12pt]{article}
\usepackage{euscript,amsmath, amssymb, amsfonts}
\usepackage{color}

\newcommand{\R}{{\cal R}}
\newcommand{\str}{{\sf sp} }
\newcommand{\be}{\begin{equation}}
\newcommand{\ee}{\end{equation}}
\newcommand{\bee}{\begin{eqnarray}}
\newcommand{\eee}{\end{eqnarray}}

\newcommand{\vv}{{\vec v}}

\newcounter{theorem}

\makeatletter \@addtoreset{theorem}{section}

\newcounter{corollary}

\makeatletter \@addtoreset{corollary}{section}

\newcounter{lemma}

\makeatletter \@addtoreset{lemma}{section}

\newcounter{proposition}

\makeatletter \@addtoreset{proposition}{section}

\newcounter{conjecture}

\makeatletter \@addtoreset{conjecture}{section}

\newcounter{remark}

\makeatletter \@addtoreset{remark}{section}

\newcounter{definition}

\makeatletter \@addtoreset{definition}{section}

\makeatletter \@addtoreset{equation}{section}

\newenvironment{proof}[1][Proof]{\noindent\textsf{#1.\ }}
{\hfill {\small $\square$}}

\begin{document}

\sloppy \title
 {
Connection between the ideals generated by traces and by supertraces
in the superalgebras of
observables of Calogero models
}

\author
 {
 S.E. Konstein%
\thanks{ I.E. Tamm Department of Theoretical Physics,
          P.N. Lebedev Physical Institute, RAS
          119991, Leninsky prosp., 53, Moscow, Russia}
\thanks{E-mail: konstein@lpi.ru}
 ,
 I.V. Tyutin$^*$%
\thanks{E-mail: tyutin@lpi.ru}     }

\date{
}

\maketitle
\thispagestyle{empty}

\begin{abstract}

If $G$ is a finite Coxeter group, then symplectic reflection algebra
$H:=H_{1,\eta}(G)$ has Lie algebra $\mathfrak {sl}_2$ of inner derivations and can be decomposed
under spin: $H=H_0 \oplus H_{1/2} \oplus H_{1} \oplus H_{3/2} \oplus ...$.
We show that if the ideals $\mathcal I_i$ ($i=1,2$) of all the vectors
from the kernel
of degenerate bilinear forms $B_i(x,y):=sp_i(x\cdot y)$,
where $sp_i$ are (super)traces on $H$, do exist, then $\mathcal I_1=\mathcal I_2$
if and only if $\mathcal I_1 \bigcap H_0=\mathcal I_2 \bigcap H_0$.

\end{abstract}


\section{Preliminaries and notation}

Let ${\mathcal A}$ be an associative superalgebra with parity $\pi$.
All expressions of linear algebra are given for homogenous elements only
and are supposed to be extended to inhomogeneous elements via linearity.

\begin{definition}\label{str}
A linear function $str$ on ${\mathcal A}$ is called a {\it supertrace} if
$$str(f \cdot g)=(-1)^{\pi(f)\pi(g)}str(g \cdot f) \ \mbox{ for all } f,g\in {\mathcal A}.$$
\end{definition}

\begin{definition}\label{tr}
A linear function $tr$ on ${\mathcal A}$ is called a {\it trace} if
$$tr(f \cdot g)=tr(g \cdot f) \ \mbox{ for all } f,g\in {\mathcal A}.$$
\end{definition}

We will use the notation $\str$ and the term "(super)trace" to denote both cases,
traces and supertraces, simultaneously.

\section{The superalgebra of observables}

Let $V={\mathbb  R}^N$ be endowed with a positive definite symmetric
bilinear form $(\cdot,\cdot)$.
For any nonzero $\vv \in V$, define the {\it
reflections} $r_\vv$ as follows:
\be\label{ref}
r_\vv: \ \ \
{\vec x} \mapsto {\vec x} -2 \frac {({\vec x},\,\vv)} {(\vv,\,\vv)} \vv \qquad
\mbox{ for any }{\vec x} \in V.
\ee

A finite set of non-zero vectors $\R\subset V$ is said to be a {\it root system}
and any vector $\vv\in \R$ is called a \emph{root}
if the following conditions hold:

i) $\R$ is ${r}_{\vec w}$-invariant for any $\vec w \in \R$,

ii) if $\vv_1,\vv_2\in \R$ are proportional to each other, then either $\vv_1=\vv_2$ or $\vv_1=-\vv_2$.

The Coxeter group $G\subset O(N, {\mathbb  R})\subset End(V)$ generated
by all reflections ${r}_\vv$ with $\vv \in \R$ is finite.

We do not apply any conditions on the scalar products of the roots because we
want to consider
both crystallographic and non-crystallographic root systems, e.g., $I_2(n)$ (see Theorem \ref{th2}).

Let $\eta$ be a complex-valued $G$-invariant function on $\R$, i.e.,
$\eta(\vv)=\eta({\vec w})$ if $r_\vv$ and $r_{\vec w}$ belong to
one conjugacy class of $G$.

We consider here
the Symplectic Reflection (Super)algebra
over complex numbers
(see \cite{sra}) $H:=H_{1,\eta}(G)$
and call it the
\emph{superalgebra of observables of Calogero model based on root system $\R$}.
\footnote
{This algebra
has a faithful
representation via Dunkl differential-difference operators $D_i$,
see \cite{Dunkl}, acting on the space of $G$-invariant smooth functions on $V$,
namely $\hat a^{\alpha}_i =\frac 1
{\sqrt{2}} (x_i + (-1)^\alpha D_i)$, see \cite{BHV,Poly}.
The Hamiltonian of the Calogero model based on the root system
\cite{Cal1, Cal2, Cal3, OP} is the operator $\hat T^{01}$  defined in (\ref{sl2})
(see \cite{BHV} ).
The wave functions are obtained in this model
via the standard Fock procedure with the Fock vacuum $ |0\rangle$ such that
$\hat a_i^0 |0\rangle$=0 for all $i$ by acting on $ |0\rangle$ with $G$-invariant polynomials
of the~$\hat a^1_i$.
}

This algebra consists of noncommuting polynomials in $2N$ indeterminates
$a^\alpha_i$, where $\alpha=0,1$ and $i=1,\, ... ,\, N$, with coefficients
in $\mathbb C[G]$ satisfying the relations (see \cite{sra} Eq. (1.15))%
\footnote{
The sign and coefficient of the sum in the rhs of Eq. (\ref{comaa})
is chosen for obtaining the Calogero model in the form
\cite{BHV}, Eq. (1), Eq. (5),  Eq. (9), Eq. (10) when $\R$ is of type
$A_{N-1}$.
}
\be\label{comaa}
[a^{\alpha}_i, a^{\beta}_j] = \varepsilon^{\alpha\beta}
\left(\delta_{ij}+
\sum_{\vv\in\R} \eta(\vv) \frac {v_i v_j}{(\vv,\,\vv)}r_{\vv}\right),
\ee
and
\be\label{comav}
r_{\vv} a^{\alpha}_i = \sum_{j=1}^N \left(\delta_{ij} - 2
\frac {v_i v_j}{(\vv,\,\vv)}\right)a^{\alpha}_j  r_{\vv}.
\ee

Here $\varepsilon^{\alpha\beta}$ is the antisymmetric tensor such that $\varepsilon^{01}=1$,
and $v_i$ ($i=1,..., N$) are the coordinates of the vector $\vv$.
The commutation relations (\ref{comaa}), (\ref{comav}) suggest
to define the {\it parity} $\pi$ by setting:
\be
\pi (a^{\alpha}_i)=1
\ \mbox{ for any }\alpha,\ i;
\qquad \pi( r_{\vv})=0 \ \mbox{ for any } \vv \in \R.
\ee
and we can consider the algebra $H$ as a superalgebra as well.


\section{$\mathfrak {sl}_2$}

Observe an important property of the superalgebra
$H$: The Lie (super)algebra of its inner
derivations
contains the Lie subalgebra $\mathfrak{sl}_2$ generated by operators
\be
D^{\alpha\beta}: \ \ \ f\mapsto D^{\alpha\beta}f=[T^{\alpha\beta}, \, f]
\ee
where $\alpha,\beta=0,1$, and $f\in H$, and polynomials $T^{\alpha\beta}$ are defined
as follows:
\be \label {sl2}
T^{\alpha\beta}:= \frac 1 2 \sum _{i=1}^N
\left( a^{\alpha}_i a^{\beta}_i + a^{\beta}_i a^{\alpha}_i \right )\, .
\ee
These operators satisfy the following relations:
\be \label{csl2}
[D^{\alpha\beta}, D^{\gamma\delta}]= \epsilon^{\alpha\gamma}
D^{\beta\delta} +\epsilon^{\alpha\delta} D^{\beta\gamma} +
\epsilon^{\beta\gamma} D^{\alpha\delta} +\epsilon^{\beta\delta}
D^{\alpha\gamma}\,,
\ee
since
\[
[T^{\alpha\beta}, T^{\gamma\delta}]= \epsilon^{\alpha\gamma}
T^{\beta\delta} +\epsilon^{\alpha\delta} T^{\beta\gamma} +
\epsilon^{\beta\gamma} T^{\alpha\delta} +\epsilon^{\beta\delta}
T^{\alpha\gamma}.
\]
It follows from Eq. (\ref{csl2}) that
the operators $D^{00}$, $D^{11}$ and $D^{01}=D^{10}$
constitute an $\mathfrak{sl}_2$-triple:
$$[D^{01},\, D^{11}] = 2 D^{11},\qquad
[D^{01},\, D^{00}] = - 2 D^{00},\qquad
[D^{11},\, D^{00}] =  - 4 D^{01}.$$

The polynomials $T^{\alpha\beta}$ commute with ${\mathbb C}[G]$, i.e.,
$[T^{\alpha\beta},\,r_\vv]=0$,
and act on the $a^{\alpha}_i$ as on vectors
of the irreducible 2-dimensional
$\mathfrak{sl}_2$-modules:
\be\label{sl2vec}
D^{\alpha\beta}a^{\gamma}_i=\left[T^{\alpha\beta},\,a^{\gamma}_i\right]=
\varepsilon^{\alpha\gamma}a^{\beta}_i +
\varepsilon^{\beta\gamma}a^{\alpha}_i, \quad \mbox{ where } i=1,\,\dots\,, N.
\ee

We will denote this $\mathfrak{sl}_2$ thus realized by the symbol $SL2$.

The subalgebra
\be
H_0:=\{f\in H \mid D^{\alpha\beta}f=0 \mbox { for any } \alpha,\, \beta \}\subset H
\ee
is called the \emph{subalgebra of singlets}.

Introduce also the subspaces $H_s:=\oplus_{i_s=1}^\infty H^{i_s}_s$,
which is the direct sum of all irreducible $SL2$-modules $H^{i_s}_s$  of
 spin $s$, for $s=0,\, 1/2,\, 1,\,...$. It is clear that $H_0$ is
 the defined above subalgebra of singlets.

The (super)algebra $H$ can be decomposed in the following way
\[
H=H_0 \oplus H_{rest},  \text{\ \ where \ \ }
H_{rest}:= H_{1/2} \oplus H_1 \oplus H_{3/2} \oplus\dots.
\]
Then each element $f\in H$
can be represented in the form $f=f_0+f_{rest}$, where $f_0 \in H_0$ and $f_{rest}\in H_{rest}$.

Note, that since $SL2$ is generated by inner derivations and $T^{\alpha\beta}$ are even elements,
each two-sided ideal $\mathcal I\subset H$ can be decomposed in an analogous way:
$\mathcal I=\mathcal I_0\oplus \mathcal I_{1/2}\oplus... $.

Since $T^{\alpha\beta}$ are even elements of the superalgebra $H$,
we have
$\str(D^{\alpha\beta}f)=0$ for any (super)trace $\str$ on $H$, and hence
the following proposition takes place%
\footnote{This elementary fact is known for a long time, see, eg, \cite{KV}.}:
\begin{proposition}\label{f0}
$\str(f)=\str(f_0)$ for any $f\in H$ and any (super)trace $\str$ on $H$.
\end{proposition}

\begin{proof}
If $s\ne 0$, then the elements of the form $D^{\alpha\beta} f $ , where
$\alpha,\, \beta =0,\,1$, and $f \in H^{i_s}_s$, $f\ne0$, span the irreducible $SL2$-module $H^{i_s}_s$.
This implies $\str f =0$ for any (super)trace on $H$ and any $f\in H_{rest}$.
\end{proof}


\section{The (super)traces on $H$}\label{trace}

It is shown in \cite{KV,KTSRA,KTroot} that the algebra $H$ has a multitude
of independent (super)traces. For the list of dimensions of the spaces of the
(super)traces
on $H_{1,\eta}(M)$ for all finite Coxeter groups $M$, see \cite{stek}.
In particular,  there is an $m$-dimensional space of
traces and an $(m+1)$-dimensional space of supertraces on $H_{1,\eta}(I_2(2m+1))$.

Every (super)trace $\str (\cdot)$ on any associative (super)algebra ${\mathcal A}$ generates
the following bilinear form on ${\mathcal A}$:
\bee\label{bf2}
B_{\str }(f,g)=\str (f\cdot g) \mbox{ for any } f,g\in {\mathcal A}.
\eee

It is obvious that if such a bilinear form $B_{\str }$ is degenerate,
then
the
kernel of this form (i.e., the set of all vectors $f \in {\mathcal A}$ such that $B_{\str}(f, g)=0$ for any
$g\in {\mathcal A}$)
is the two-sided
ideal ${\mathcal I^{\str}}\subset {\mathcal A}$.

The ideals of this sort are found, for example,
in \cite[Theorem 9.1]{2m+1}  (generalizing the results of \cite{V1, V2}  and \cite{K2} for the two-
and three-particle Calogero models).

Theorem 9.1 from \cite{2m+1} may be shortened to the following theorem:
\vskip 2mm
\begin{theorem}\label{th2}
Let $ m\in \mathbb Z$, where $m\geqslant 1$
and $n=2m+1$.
Then

\emph{1)} The associative algebra $H_{1,\eta}(I_2(n))$
has  nonzero  traces $tr_\eta$ such that the symmetric invariant bilinear form
$B_{tr_\eta}(x,y)=tr_\eta(x\cdot y)$  is degenerate if and only if
$\eta= \frac z n$, where $z\in \mathbb Z \setminus n\mathbb Z$.
All such traces are proportional to each other.

\emph{2)} The associative superalgebra $H_{1,\eta}(I_2(n))$
has  nonzero supertraces $str_\eta$ such that the supersymmetric invariant bilinear form
$B_{str_\eta}(x,y)=str_\eta(x\cdot y)$  is degenerate if
$\eta= \frac z n$, where $z\in \mathbb Z \setminus n\mathbb Z$.
All such supertraces are proportional to each other.

\emph{3)} The associative superalgebra $H_{1,\eta}(I_2(n))$
has  nonzero supertraces $str_{\eta}$ such that the supersymmetric invariant bilinear form
$B_{str_{\eta}}(x,y)=str_{\eta}(x\cdot y)$  is degenerate if
$\eta= z + \frac 1 2$, where $z\in \mathbb Z$.
All such supetraces are proportional to each other.

\emph{4)} For all other values of $\eta$, all nonzero traces and supertraces are nondegenerate.
\end{theorem}

Theorem \ref{th2} implies that if $z\in \mathbb Z \setminus n\mathbb Z$, then
there exists the degenerate trace $tr_z$ generating the ideal $\mathcal I^{tr_z}$ consisting of the kernel
of the degenerate form $B_{tr_z}(f,g)=tr_z(f\cdot g)$, and simultaneously
the degenerate supertrace $str_z$ generating the ideal $\mathcal I^{str_z}$ consisting of the kernel
of the degenerate form $B_{str_z}(f,g)=str_z(f\cdot g)$.

A question arises: is it true that $\mathcal I^{tr_z}=\mathcal I^{str_z}$?

Answer to this and other similar questions
can be considerably simplified by considering only the singlet parts of these ideals.

The following theorem justifies this method:

\vskip 2mm
\begin{theorem}\label{th3}
Let $\str_1$ and $\str_2$ be degenerate (super)traces on
$H$.
They generate the
two-sided ideals $\mathcal I_1$ and $\mathcal I_2$ consisting of the kernels
 of bilinear forms $B_1 (f,g)=\str_1(f\cdot g)$ and
$B_2 (f,g)=\str_2(f\cdot g)$, respectively.

Then
$\mathcal I_{1}=\mathcal I_{2}$ if and only if
$\mathcal I_{1} \bigcap H_0 = \mathcal I_{2} \bigcap H_0$.
\end{theorem}

\begin{proof}
It suffices to prove that
if
$\mathcal I_{1} \bigcap H_0 = \mathcal I_{2} \bigcap H_0$,
then
$\mathcal I_{1}=\mathcal I_{2}$.

Consider any non-zero element $f\in \mathcal I_1$.
For any $g\in H$, we have $\str_1(f\cdot g)=0$,
$f\cdot g\in \mathcal I_1$ and $(f\cdot g)_0\in \mathcal I_1$.
So
$(f\cdot g)_0\in \mathcal I_1 \bigcap H_0$.
Due to hypotheses of this Theorem,
$(f\cdot g)_0\in \mathcal I_2 \bigcap H_0$,
and hence ${\str_2((f\cdot g)_0)=0}$.
 Proposition \ref{f0} gives
$\str_2(f\cdot g)=\str_2((f\cdot g)_0)$
which implies
$\str_2(f\cdot g)=0$.

Therefore, $f\in \mathcal I_2$.
\end{proof}


\section*{Acknowledgments}
The authors (S.K. and I.T.) are grateful to Russian Fund for Basic Research
(grant No.~${\text{17-02-00317}}$)
for partial support of this work.


\end{document}